\documentclass{amsart}
\usepackage{amssymb}
\input xy
\xyoption{all}

\newcommand{\C}{\mathcal C}

\newcommand{\N}{\mathbb N}
\newcommand{\R}{\mathbb R}

\newcommand{\F}{\mathcal F}

\newtheorem{theorem}{Theorem}

\newtheorem{lemma}{Lemma}

\newtheorem{problem}{Problem}
\newtheorem{proposition}{Proposition}
\newtheorem{df}{Definition}
\newtheorem{ex}{Example}

\begin{document}

\title{Equilibrium for max-plus payoff}

\author{Taras Radul}

\maketitle

Institute of Mathematics, Casimirus the Great University of Bydgoszcz, Poland;
\newline
Department of Mechanics and Mathematics, Ivan Franko National University of Lviv,
Universytetska st., 1. 79000 Lviv, Ukraine.
\newline
e-mail: tarasradul@yahoo.co.uk



\begin{abstract} We study equilibrium concepts in non-cooperative games under uncertainty where both beliefs and mixed strategies are represented by non-additive measures (capacities). In contrast to the classical Nash framework based on additive probabilities and linear convexity, we employ capacities and max-plus integrals to model qualitative and idempotent decision criteria. Two equilibrium notions are investigated: Nash equilibrium in mixed strategies expressed by capacities, and equilibrium under uncertainty in the sense of Dow and Werlang, where players choose pure strategies but evaluate payoffs with respect to non-additive beliefs. For games with compact strategy spaces and continuous payoffs, we establish existence results for both equilibrium concepts using abstract convexity techniques and a Kakutani-type fixed point theorem.
\end{abstract}


\section{Introduction}

Classical Nash equilibrium theory relies on fixed point arguments formulated within the framework of linear convexity. In this setting, mixed strategies are modeled as probability measures, that is, additive measures defined on the sets of pure strategies. Over the last decades, however, there has been a growing interest in extending Nash equilibrium theory beyond this classical linear structure. For example, Aliprantis, Florenzano and Tourky \cite{AFT} worked in ordered topological vector spaces, Luo \cite{L} studied games on topological semilattices, while Vives \cite{Vi} considered strategic interactions on complete lattices. The existence of Nash equilibria in the context of idempotent convexity was established by Briec and Horvath \cite{Ch}.

The use of additive probability measures presupposes precise knowledge of the likelihood of all relevant events. In many contemporary economic and decision-theoretic models this assumption is unrealistic. Decision-making under uncertainty often involves situations in which probabilities are unknown or only vaguely specified. To address this issue, Gilboa \cite{Gil} and Schmeidler \cite{Sch} proposed axiomatic foundations for expectations represented by Choquet integrals with respect to non-additive measures, known as capacities or fuzzy measures.

Dow and Werlang \cite{DW} applied this approach to two-player games in which players hold non-additive beliefs, represented by convex capacities, about their opponents' actions, while still choosing pure strategies. Payoffs were evaluated using Choquet integrals, and an appropriate equilibrium concept was introduced together with an existence result. This framework was extended to games with finitely many players in \cite{EK}. Another line of research on games with Choquet payoff functions can be found in \cite{Ma}. In all these works, the sets of pure strategies were assumed to be finite.

An alternative to the Choquet expected utility model is provided by qualitative decision theory, where expectations are expressed via the Sugeno integral. This approach has been extensively studied in recent years (see, for example, \cite{DP}, \cite{DP1}, \cite{CH1}, \cite{CH}). Unlike the Choquet integral, which performs an averaging operation, the Sugeno integral selects a median-type value of utilities and thus reflects a purely qualitative aggregation mechanism.

The equilibrium notion introduced in \cite{DW} and \cite{EK} was adapted to games with payoff functions defined by the Sugeno integral in \cite{R4}, and further generalized to t-normed integrals in \cite{R5}.

Kozhan and Zarichnyi \cite{KZ}, as well as Glycopantis and Muir \cite{GM}, considered games in which players not only form non-additive beliefs about their opponents but also employ mixed strategies expressed by capacities. An analogous framework for games with Sugeno payoff functions was studied in \cite{R3} and generalized to t-normed integrals in \cite{R6}. Games with possibility-valued strategies were investigated by Hosni and Marchioni in \cite{HM} and \cite{HM1}, while further existence results for equilibria with Sugeno payoffs were obtained in \cite{Pet}.

In this paper we focus on games with payoff functions represented by the max-plus integral, a fuzzy integral based on the maximum and addition operations, introduced in \cite{R2}. This integral provides a natural link between capacity theory and idempotent (max-plus) analysis. We analyze both equilibrium concepts mentioned above - Nash equilibrium in mixed strategies and equilibrium under uncertainty - for games with max-plus payoff functions.

The paper is organized as follows. Section 2 recalls basic notions concerning capacities and the max-plus integral. In Section 3 we study Nash equilibria in games where players are allowed to use mixed non-additive strategies represented by capacities. Section 4 is devoted to equilibrium under uncertainty in the sense of Dow and Werlang and to a comparison between the two equilibrium concepts. Final remarks and conclusions are collected in Section 5.

\section{Capacities and max-plus integral:preliminaries}

Throughout the paper all spaces are assumed to be compact Hausdorff spaces (compacta), except for $\R$, and all maps are assumed to be continuous. By $\F(X)$ we denote the family of all closed subsets of a compactum $X$. The Banach space of continuous real-valued functions on $X$, endowed with the supremum norm, is denoted by $C(X)$. For $c\in\R$, the constant function on $X$ with value $c$ is denoted by $c_X$. The pointwise lattice operations $\vee$ and $\wedge$ are used on $C(X)$.

We begin by recalling the notion of a capacity on a compactum $X$. Our terminology follows \cite{NZ}.

\begin{df}
A function $\nu:\F(X)\to [0,1]$ is called an {\it upper-semicontinuous capacity} on $X$ if the following conditions hold for all closed subsets $F,G\subset X$:
\begin{enumerate}
\item $\nu(X)=1$ and $\nu(\emptyset)=0$;
\item $F\subset G$ implies $\nu(F)\le \nu(G)$;
\item if $\nu(F)<a$ for some $a\in[0,1]$, then there exists an open set $O\supset F$ such that $\nu(B)<a$ for every compactum $B\subset O$.
\end{enumerate}
\end{df}

For a singleton $\{a\}\subset X$ we write simply $\nu(a)$ instead of $\nu(\{a\})$.
Following \cite{NZ}, each capacity $\nu$ is extended to open subsets $U\subset X$ by
$$
\nu(U)=\sup\{\nu(K)\mid K\subset U,\ K\ \text{closed}\}.
$$

It was shown in \cite{NZ} that the space $MX$ of all upper-semicontinuous capacities on $X$ is compact when equipped with the topology generated by the subbase consisting of sets of the form
$$
O_-(F,a)=\{\nu\in MX\mid \nu(F)<a\},
$$
where $F$ is closed and $a\in[0,1]$, and
$$
O_+(U,a)=\{\nu\in MX\mid \nu(U)>a\},
$$
where $U$ is open and $a\in[0,1]$. Since all capacities considered in this paper are upper-semicontinuous, we shall simply refer to elements of $MX$ as capacities.

\begin{df}\label{pos}
A capacity $\nu\in MX$ is called a {\it necessity capacity} (respectively, a {\it possibility capacity}) if for every family $\{A_t\}_{t\in T}$ of closed subsets of $X$ such that $\bigcup_{t\in T}A_t$ is closed, the following condition holds:
$$
\nu\Bigl(\bigcap_{t\in T}A_t\Bigr)=\inf_{t\in T}\nu(A_t)
\quad
\Bigl(\nu\Bigl(\bigcup_{t\in T}A_t\Bigr)=\sup_{t\in T}\nu(A_t)\Bigr).
$$
\end{df}

(See \cite{WK} for more details.)

We denote by $NX$ and $\Delta X$ the subspaces of $MX$ consisting of necessity and possibility capacities, respectively. Since $X$ is compact and $\nu$ is upper-semicontinuous, a capacity $\nu$ belongs to $NX$ if and only if it satisfies $\nu(A\cap B)=\min\{\nu(A),\nu(B)\}$ for all closed $A,B\subset X$.

Given a capacity $\nu$ on $X$, its dual capacity $\kappa X(\nu)$ is defined by
$$
\kappa X(\nu)(F)=1-\nu(X\setminus F),\quad F\in\F(X).
$$
The mapping $\kappa X:MX\to MX$ is a homeomorphism and an involution \cite{NZ}. Moreover, $\nu$ is a necessity capacity if and only if $\kappa X(\nu)$ is a possibility capacity. Consequently, $\nu\in\Delta X$ if and only if $\nu(A\cup B)=\max\{\nu(A),\nu(B)\}$ for all closed $A,B\subset X$. Both $NX$ and $\Delta X$ are closed subspaces of $MX$.

We now recall the definition of the max-plus integral. For $\varphi\in C(X)$ and $t\in\R$ we set $\varphi_t=\{x\in X\mid \varphi(x)\ge t\}$. We adopt the convention $\ln(0)=-\infty$ and $-\infty+\gamma=-\infty$ for all $\gamma\in\R$.

\begin{df}\cite{R2}
Let $\varphi\in C(X)$ and $\nu\in MX$. The max-plus integral of $\varphi$ with respect to $c$ is defined by
$$
\int_X^{\vee+} \varphi\, d\nu
=
\max\{\ln(\nu(\varphi_t))+t\mid t\in\R\}.
$$
\end{df}

The maximum in the above formula exists due to the upper-semicontinuity of the capacity $\nu$. It is sufficient to consider $t\in[\min\varphi(X),\max\varphi(X)]$.

For $\psi\in C(X)$ we define a function $lX^\psi:MX\to\R$ by
$$
lX^\psi(\nu)=\int_X^{\vee+} \psi\, d\nu.
$$

\begin{lemma}\cite{R1}\label{contint} The map $lX^\psi$ is continuous for each $\psi\in C(X)$.
\end{lemma}

The tensor product of probability measures is a classical and powerful construction, widely used in the study of spaces of probability measures on compacta (see, for example, Chapter~8 in \cite{FF}).
An analogue of this operation for capacities was introduced in \cite{KZ} and is based on the monad structure of the capacity functor.
An explicit formula for the tensor product of capacities, avoiding the categorical formalism, was later obtained in \cite{R4}.

For $\mu_1\in MX_1$, $\mu_2\in MX_2$, and $B\in\F(X_1\times X_2)$, the tensor product is defined by
$$
\mu_1\otimes\mu_2(B)
=
\max\Bigl\{
t\in[0,1]\ \Big|\
\mu_1\bigl(\{x\in X_1\mid \mu_2(p_2((\{x\}\times X_2)\cap B))\ge t\}\bigr)\ge t
\Bigr\}.
$$
Note that this formula can equivalently be written in the form
$$
\mu_1\otimes\mu_2(B)
=
\max\Bigl\{
\mu_1\bigl(\{x\in X_1\mid \mu_2(p_2((\{x\}\times X_2)\cap B))\ge t\}\bigr)\wedge t
\ \Big|\ t\in[0,1]
\Bigr\}.
$$

The problem of multiplying capacities has been extensively studied in possibility theory and its applications to game theory and decision theory. In this context, the notions of a joint possibility distribution (see, for example, \cite{HM}) and aggregation of capacities \cite{Df} are commonly used. As discussed in \cite{R6}, the differences between these approaches are largely terminological.

It was observed in \cite{KZ} that the definition of the tensor product can be extended inductively to an arbitrary finite number of factors.

The following lemma was proved in \cite{R4}.

\begin{lemma}\label{P}
Let $X_i$ be a compactum, $A_i\in\F(X_i)$, and $\mu_i\in MX_i$ be such that $\mu_i(X_i\setminus A_i)=0$ for each $i\in\{1,\dots,n\}$.
Then
$$
\otimes_{i=1}^n\mu_i\Bigl(\prod_{i=1}^n X_i\setminus\prod_{i=1}^n A_i\Bigr)=0.
$$
\end{lemma}

We equip the space $MX$ of all capacities with the natural partial order: for $\nu,\mu\in MX$ we write $\nu\le\mu$ if and only if $\nu(A)\le\mu(A)$ for every $A\in\F(X)$.
The following lemma shows that the tensor product is monotone with respect to each coordinate.

\begin{lemma}\label{admis}
Let $\mu_1\in MX_1,\dots,\nu_i,\mu_i\in MX_i,\dots,\mu_n\in MX_n$ and assume that $\nu_i\le\mu_i$.
Then
$$
\mu_1\otimes\dots\otimes\nu_i\otimes\dots\otimes\mu_n
\le
\mu_1\otimes\dots\otimes\mu_i\otimes\dots\otimes\mu_n.
$$
\end{lemma}

The proof of Lemma~\ref{admis} follows directly from the monotonicity of capacities.

The next lemma shows that the max-plus integral is also monotone with respect to the underlying capacity. Its proof is again an immediate consequence of the monotonicity property.

\begin{lemma}\label{monot}
Let $\nu,\mu\in MX$ be such that $\nu\le\mu$.
Then
$$
\int_X^{\vee+} \varphi\, d\nu
\le
\int_X^{\vee+} \varphi\, d\mu
$$
for every $\varphi\in C(X)$.
\end{lemma}

\section{Nash equilibrium in mixed strategies}\label{NashEq}

We recall the notion of Nash equilibrium and some existence results. Consider an $n$-player game
$u:X=\prod_{j=1}^n X_j\to\R^n$,
where each $X_i$ is a compact Hausdorff space. The coordinate function $u_i:X\to\R$ represents the payoff of player $i$.
For $x\in X$ and $t_i\in X_i$ we write
$(x;t_i)=(x_1,\dots,x_{i-1},t_i,x_{i+1},\dots,x_n)$.

A point $x\in X$ is called a Nash max-equilibrium (respectively, min-equilibrium) if for every $i\in\{1,\dots,n\}$ and every $t_i\in X_i$,
$$
u_i(x;t_i)\le u_i(x)
\quad
(u_i(x;t_i)\ge u_i(x)).
$$

Now we extend the game
$u:X=\prod_{j=1}^n X_j\to\R^n$
to a game in mixed strategies
$eu:\prod_{j=1}^n MX_j\to\R^n$
by means of the max-plus integral and the tensor product of capacities.

For $(\nu_1,\dots,\nu_n)\in\prod_{j=1}^n MX_j$ we define the expected payoff functions
$eu_i:\prod_{j=1}^n MX_j\to\R$ by
$$
eu_i(\nu_1,\dots,\nu_n)
=
\int_X^{\vee+} u_i\, d(\nu_1\otimes\dots\otimes\nu_n).
$$

Lemma~\ref{contint} together with the continuity of the tensor product implies the following result.

\begin{lemma}\label{cont}
The function $eu_i$ is continuous for each $i\in\{1,\dots,n\}$.
\end{lemma}

We now discuss the existence of Nash equilibria in mixed strategies represented by capacities.
There exists a trivial solution to this problem.
Each space $MX_i$ contains the greatest element $\mu_i$, defined by
$$
\mu_i(A)=
\begin{cases}
0,& A=\emptyset,\\
1,& A\neq\emptyset,
\end{cases}
\qquad A\in\F(X_i).
$$
Lemmas~\ref{admis} and~\ref{monot} immediately imply the following theorem.

\begin{theorem}
$(\mu_1,\dots,\mu_n)$ is a Nash max-equilibrium point.
\end{theorem}

Since each space $MX_i$ also contains the smallest element, an analogous trivial solution exists for the min-equilibrium problem.

We now restrict attention to the subspace of possibility capacities $\Delta X\subset MX$.
It is easy to see that $\Delta X$ contains the greatest element described above, but does not contain the smallest one.
Consequently, the existence problem for a min-equilibrium in $\Delta X$ admits no trivial solution.
This problem is investigated in the present section.

Thus, we consider the game in mixed strategies
$r:\prod_{j=1}^n \Delta X_j\to\R^n$,
where $r=eu|\prod_{j=1}^n\Delta X_j$.
Its coordinate functions
$r_i:\prod_{j=1}^n \Delta X_j\to\R$
are the restrictions of $eu_i$ to $\prod_{j=1}^n \Delta X_j$.

 To establish the existence of a Nash equilibrium, we employ a non-linear convexity structure.
A family $\mathcal C$ of closed subsets of a compactum $X$ is called a {\it convexity} on $X$ if $\mathcal C$ is closed under intersections and contains both $X$ and the empty set.
The elements of $\mathcal C$ are called $\mathcal C$-convex (or simply convex).
Although we follow the general concept of abstract convexity from \cite{vV}, our definition is different.
We consider only closed convex sets; such structures are called closure structures in \cite{vV}.
Our definition coincides with that used in \cite{W}.
The full family of convex sets in the sense of \cite{vV} can be obtained by taking unions of up-directed families.
In what follows, we assume that every convexity contains all singletons.

A convexity $\mathcal C$ on $X$ is called $T_4$ (or normal) if for any disjoint sets $C_1,C_2\in\mathcal C$ there exist $S_1,S_2\in\mathcal C$ such that
$S_1\cup S_2=X$,
$C_1\cap S_2=\emptyset$,
and $C_2\cap S_1=\emptyset$
(see, for example, \cite{RZ}).

Let $\mathcal C_i$ be a convexity on $X_i$.
A function
$f:X=\prod_{j=1}^n X_j\to\R$
is called quasiconcave (respectively, quasiconvex) with respect to the $i$-th variable if
$$
(f_i^x)^{-1}([t,+\infty))\in\mathcal C_i
\quad
\bigl((f_i^x)^{-1}((-\infty,t])\in\mathcal C_i\bigr)
$$
for every $t\in\R$ and every $x\in X$,
where $f_i^x:X_i\to\R$ is defined by
$f_i^x(z_i)=f_i(x;z_i)$ for $z_i\in X_i$.

\begin{theorem}\label{NN}\cite{R3} Let $f:X=\prod_{j=1}^n X_j\to\R^n$ be a game with a  normal convexity  $\C_i$ defined  on each compactum $X_i$ such that all convex sets are connected, the function $f$ is continuous  and the function $f_i:X\to\R$ is quasiconcave (quasiconvex) with respect to the $i$-th variable for each $i\in\{1,\dots,n\}$. Then there exists a Nash max-equilibrium (min-equilibrium) point.
\end{theorem}

We note that in \cite{R3} this theorem was proved only for the max-equilibrium.
However, the same proof applies verbatim to the min-equilibrium case.

We now turn to the existence of Nash equilibria in mixed strategies represented by possibility capacities.
There exists a trivial solution for the max-equilibrium, since the greatest element of each $MX_i$ belongs to $\Delta X_i$.
In contrast, no such trivial solution exists for the min-equilibrium, as $\Delta X_i$ does not contain the smallest element.

To establish the existence of a min-equilibrium, we introduce a suitable convexity structure on $\Delta X$.
We use an idempotent convexity considered in \cite{Ch} for finite-dimensional spaces, where it was called {\it B-convexity}.

The B-convexity on $\Delta X$ for a compactum $X$ was studied in \cite{R6}.
For $\nu,\mu\in\Delta X$ and $s\in[0,1]$, define
$$
(s\cdot\nu\vee\mu)(A)=s\cdot\nu(A)\vee\mu(A),
\qquad A\in\F(X).
$$
It is easy to verify that $s\cdot\nu\vee\mu\in\Delta X$.
A subset $C\subset\Delta X$ is called {\it B-convex} if for every $\nu,\mu\in\Delta X$ and every $s\in[0,1]$ we have $s\cdot\nu\vee\mu\in C$.
Clearly, every B-convex set is connected.
Thus, we consider on $\Delta X$ the convexity $\mathcal C_X$ consisting of all closed B-convex subsets of $\Delta X$.
It was shown in \cite{R6} that the convexity $\mathcal C_X$ is normal for every compactum $X$.

\begin{lemma}\label{QC}
For each $i\in\{1,\dots,n\}$, the map
$r_i:\prod_{j=1}^n \Delta X_j\to[0,1]$
is quasiconvex with respect to the $i$-th variable.
\end{lemma}

\begin{proof} We prove the lemma for the case $n=2$; the general case is analogous.
Without loss of generality, assume $i=1$.
Fix $s\in\R$ and $\mu\in\Delta X_2$.
We show that the set $(r_1^\mu)^{-1}((-\infty,s])$ is B-convex.

Take $\nu_1,\nu_2\in\Delta X_1$ such that
$r_1(\nu_j,\mu)\le s$ for $j=1,2$, and let $c\in[0,1]$.
For $t\in\R$ set $U_t=(u_1)_t$.
Then
$$r_1(c\cdot\nu_1\vee\nu_2,\mu)=\int_X^{\vee+} u_1 d((c\cdot\nu_1\vee\nu_2)\otimes\mu)=$$
$$=\max\{\ln(((c\cdot\nu_1\vee\nu_2)\otimes\mu)(U_t))+ t\mid t\in\R\}=$$
$$=\max\{\ln(\max \{(c\cdot\nu_1\vee\nu_2)\{x\in X_1|\mu(p_2((\{x\}\times X_2)\cap U_t))\ge \tau\})\wedge \tau\mid \tau\in[0,1]\}+ t\mid t\in\R\}=$$
$$=\max\{\ln(\max \{(c\cdot\nu_1\vee\nu_2)\{x\in X_1|\mu(p_2((\{x\}\times X_2)\cap U_t))\ge \tau\})\wedge \tau\mid \tau\in[0,1]\}+ t\mid t\in\R\}\le$$
$$\le\max\{\ln(\max \{\nu_1\{x\in X_1|\mu(p_2((\{x\}\times X_2)\cap U_t))\ge \tau\})\wedge \tau\mid \tau\in[0,1]\}+ t\mid t\in\R\}\vee$$
$$\vee\max\{\ln(\max \{\nu_2\{x\in X_1|\mu(p_2((\{x\}\times X_2)\cap U_t))\ge \tau\})\wedge \tau\mid \tau\in[0,1]\}+ t\mid t\in\R\}=$$
$$=r_1(\nu_1,\mu)\vee r_1(\nu_2,\mu)\le s.$$
\end{proof}

Theorem~\ref{NN} together with Lemma~\ref{QC} yields the following result.

\begin{theorem}\label{NMP} There exists a Nash min-equilibrium  point for the game with expected payoff functions $r_i:\prod_{j=1}^n \Delta X_j\to[0,1]$.
\end{theorem}

\section{Equilibrium under uncertainty} In this section we consider an equilibrium concept introduced in \cite{DW,EK}, where players choose pure strategies while beliefs about opponents' actions are represented by capacities.

Let $p:X=\prod_{i=1}^n X_i\to\R^n$ be an $n$-player game, where each strategy space $X_i$ is a compact Hausdorff space.
We assume that the payoff function $p$ is continuous.
For $i\in\{1,\dots,n\}$ we denote by
$X_{-i}=\prod_{j\neq i} X_j$
the space of strategy profiles of all players other than player~$i$.
For $x\in X$, the corresponding element of $X_{-i}$ is denoted by $x_{-i}$.
In contrast to the standard game-theoretic framework, the beliefs of player~$i$ about the behavior of the opponents are represented by capacities on $X_{-i}$.

For each $i\in\{1,\dots,n\}$ we define the expected payoff function
$P_i:X_i\times MX_{-i}\to\R$ by
$$
P_i(x_i,\nu)
=
\int_{X_{-i}}^{\vee+} p_i^{x_i}\, d\nu,
$$
where $p_i^{x_i}:X_{-i}\to\R$ is given by
$p_i^{x_i}(y)=p_i(x_i,y)$,
$x_i\in X_i$ and $\nu\in MX_{-i}$.

To establish continuity of the functions $P_i$, we introduce some notation and recall a technical lemma from \cite{R4}.
Let $f:X\times Y\to[0,1]$ be a function.
For $x\in X$ and $t\in[0,1]$, set
$f^x_{\le t}=\{y\in Y\mid f(x,y)\le t\}$.
Analogously, we define the sets
$f^x_{\ge t}$, $f^x_{<t}$, and $f^x_{>t}$.
Note that in the special case $X=\{x\}$ we previously used the shorter notation $f_t=f^x_{\ge t}$.

\begin{lemma}\cite{R4}\label{C0}
Let $f:X\times Y\to\R$ be a continuous function on the product of compacta $X$ and $Y$.
Then for every $x\in X$, $t\in\R$, and $\delta>0$ there exists an open neighborhood $O$ of $x$ such that
$f^z_{\le t}\subset f^x_{<t+\delta}$
(respectively,
$f^z_{\ge t}\subset f^x_{>t-\delta}$)
for all $z\in O$.
\end{lemma}

\begin{lemma}\label{C} The map $P_i$ is continuous.
\end{lemma}

\begin{proof} Fix $x\in X_i$ and $\nu_0\in MX_{-i}$, and suppose that $P_i(x,\nu_0)=s<a$ for some $a\in\R$.
Set $\varepsilon=a-s>0$.
Let $c=\min p_i(X)$ and $b=\max p_i(X)$.
Clearly, $P_i(X_i\times MX_{-i})\subset[c,b]$.
If $c=b$, the claim is trivial, so assume $c<b$.

Choose $n\in\N$ such that $(b-c)/n<\varepsilon/4$ and define
$t_k=c+k(b-c)/n$ for $k\in\{0,\dots,n\}$.
Then $t_0=c$ and $t_n=b$.
Applying the second part of Lemma~\ref{C0} to the continuous function
$f=p_i:X_i\times X_{-i}\to\R$
with $\zeta=(b-c)/n$, we obtain an open neighborhood $O$ of $x$ such that
$f^z_{\ge t_{k+1}}\subset f^x_{\ge t_k}$
for all $z\in O$ and $k\in\{0,\dots,n-1\}$.

Choose $\delta>0$ such that
$$
\ln(\nu_0(f^x_{\ge t_k})+\delta)
<
\ln(\nu_0(f^x_{\ge t_k}))+\varepsilon/4
$$
for all $k$ with $\nu_0(f^x_{\ge t_k})>0$.
Define
$$
V=\{\nu\in MX_{-i}\mid
\nu(f^x_{\ge t_k})<\nu_0(f^x_{\ge t_k})+\delta
\text{ for all }k\in\{0,\dots,n\}\}.
$$
Then $V$ is a neighborhood of $\nu_0$.

For any $(z,\nu)\in O\times V$ and $t\in[c,b]$ we consider three cases.

If $t\le t_1$, then
$$
\ln(\nu(f^z_{\ge t}))+t
\le c+\varepsilon/4
\le P_i(x,\nu_0)+\varepsilon/4.
$$

If $t\in(t_k,t_{k+1}]$ for some $k\ge1$ and $\nu(f^z_{\ge t})>0$, then
$$\ln(\nu(f^z_{\ge t}))+t\le \ln(\nu(f^z_{\ge t_k}))+t\le\ln(\nu(f^x_{\ge t_{k-1}}))+t\le\ln(\nu_0(f^x_{\ge t_{k-1}})+\delta)+t<$$ $$<\ln(\nu_0(f^x_{\ge t_{k-1}}))+\varepsilon/4+t_{k-1}+\varepsilon/2\le P_i(x,\nu_0)+3\varepsilon/4.$$

Finally, if $\nu(f^z_{\ge t})=0$, then
$\ln(\nu(f^z_{\ge t}))+t=-\infty<P_i(x,\nu_0)$.

Thus $P_i(z,\nu)<a$ for all $(z,\nu)\in O\times V$.

Let $P_i(x,\nu_0)=s>a$ for some $a\in\R$. There exists $t_0\in \R$ such that $P_i(x,\nu_0)=\ln(\nu_0(f^x_{\ge t_0}))+t_0$ (here as above we denote $f=p_i$). Put  $\varepsilon=s-a>0$.

Applying  the first part of Lemma \ref{C0}  to the continuous function $f:X_i\times X_{-i}\to\R$, $t=t_0-\varepsilon/2$ and $\delta=\varepsilon/4$ we can choose a neighborhood $O$ of $x$ such that for each $z\in O$ we have $f^z_{\le t_0-\varepsilon/2}\subset f^x_{< t_0-\varepsilon/4}$. Going to the complement, we obtain $f^z_{\ge t_0-\varepsilon/2}\supset f^z_{> t_0-\varepsilon/2}\supset f^x_{\ge t_0-\varepsilon/4}\supset f^x_{> t_0-\varepsilon/4}$.

Choose $\zeta>0$ such that $\ln(\nu_0(f^x_{\ge t_0})-\zeta)>\ln(\nu_0(f^x_{\ge t_0}))-\varepsilon/2$.   Put $V=\{\nu\in M(X_{-i})\mid \nu(f^x_{> t_0-\varepsilon/4})>\nu_0(f^x_{\ge t_0})-\zeta\}$, then $V$ is a neighborhood of $\nu_0$.

Consider any $(z,\nu)\in O\times V$. Then we have $$\ln(\nu(f^z_{\ge t_0-\varepsilon/2}))+t_0-\varepsilon/2\ge \ln(\nu(f^x_{> t_0-\varepsilon/4}))+t_0-\varepsilon/2>\ln(\nu_0(f^x_{\ge t_0})-\zeta)+t_0-\varepsilon/2>$$ $$>\ln(\nu_0(f^x_{\ge t_0}))-\varepsilon/2+t_0-\varepsilon/2= P_i(x,\nu_0)-\varepsilon=a.$$

Thus $P_i(z,\nu)\ge \ln(\nu(f^z_{\ge t_0-\varepsilon/2}))+t_0-\varepsilon/2>a$ for each $(z,\nu)\in O\times V$. Hence the map $P_i$ is continuous.
\end{proof}

For $\nu_i\in M(X_{-i})$ denote by $R_i(\nu_i)=\{x\in X_i\mid P_i(x,\nu_i)=\max\{P_i(z,\nu_i)\mid z\in X_i\}$ the best response correspondence
of player $i$ given belief $\nu_i$. By Lemma~\ref{C}, the set $R_i(\nu_i)$ is well defined and compact.

A system of beliefs $(\nu_1,\dots,\nu_n)$ with $\nu_i\in MX_{-i}$ is called
an {\it equilibrium under uncertainty} if
$$
\nu_i\bigl(X_{-i}\setminus\prod_{j\neq i}R_j(\nu_j)\bigr)=0
\quad\text{for all }i.
$$
This means that each player believes that the opponents choose best responses.
Equilibria under uncertainty were studied for Choquet payoffs in \cite{DW,EK},
for Sugeno payoffs in \cite{R4}, and for $t$-normed integrals in \cite{R5}.

The following example illustrates the difference between Nash equilibrium
and equilibrium under uncertainty.

\begin{ex}\label{ex}
Consider the two-player game
$p:\{a,b\}\times\{a,b\}\to\R^2$
defined by
$p_1(a,a)=p_1(a,b)=1$ and $p_1(b,a)=p_1(b,b)=0$,
while $p_2$ is arbitrary.
Let $\nu\in MX(\{a,b\})$ be the greatest capacity.
 Then the pair $(\nu,\nu)$ is  an Nash  equilibrium for the game with expected payoff functions $ep_i$ as was discussed in Section \ref{NashEq}. On the other hand we have $P_1(a,\nu)=1$ and $P_1(b,\nu)=0$. Hence $R_1(\nu)=\{a\}$ and $(\nu,\nu)$ is not  an equilibrium under uncertainty.
\end{ex}

Next, we prove the existence of an equilibrium under uncertainty. Moreover, we show that each belief can be represented as a tensor product of possibility capacities on the spaces $X_i$. For this purpose, we require a suitable fixed-point theorem.

By a multimap (or set-valued map) from a set $X$ into a set $Y$ we mean a mapping $F:X\to 2^Y$, denoted by $F:X\multimap Y$. If $X$ and $Y$ are topological spaces, then a multimap $F:X\multimap Y$ is called \emph{upper semicontinuous} (USC) if for every open set $O\subset Y$ the set
\[
\{x\in X\mid F(x)\subset O\}
\]
is open in $X$. It is well known that a multimap between compacta with closed values is USC if and only if its graph is closed in $X\times Y$.

Let $F:X\multimap X$ be a multimap. A point $x\in X$ is called a fixed point of $F$ if $x\in F(x)$. The following result is a counterpart of Kakutani's fixed-point theorem in the setting of abstract convexity and is a special case of Theorem~3 from \cite{W}.

\begin{theorem}\label{KA} Let $\C$ be a normal convexity on a compactum $X$ such that all convex sets are connected and $F:X\multimap X$ is a USC multimap with values in $\C\setminus\{\emptyset\}$. Then $F$ has a fixed point.
\end{theorem}

We consider on the space of possibility capacities $\Delta X$ the normal B-convexity  $\C_X$ described in Section \ref{NashEq}. Recall that all elements of $\mathcal{C}_X$ are connected.

\begin{theorem}\label{Exi}
There exists a system $(\mu_1,\dots,\mu_n)\in \Delta X_1\times\cdots\times \Delta X_n$ such that $(\mu_1^\ast,\dots,\mu_n^\ast)$ is an equilibrium under uncertainty with max-plus payoff, where
\[
\mu_i^\ast=\bigotimes_{j\neq i}\mu_j.
\]
\end{theorem}

\begin{proof}
For each $i\in\{1,\dots,n\}$ define a multimap
\[
\gamma_i:\prod_{j=1}^n\Delta X_j \multimap \Delta X_i
\]
by
\[
\gamma_i(\mu_1,\dots,\mu_n)=
\{\mu\in\Delta X_i\mid \mu(X_i\setminus R_i(\mu_i^\ast))=0\},
\]
where $\mu_i^\ast=\bigotimes_{j\neq i}\mu_j$.
By the definition of the topology on $\Delta X_i$, the set
$\gamma_i(\mu_1,\dots,\mu_n)$ is closed for every
$(\mu_1,\dots,\mu_n)$, and it is immediate that
$\gamma_i(\mu_1,\dots,\mu_n)\in\mathcal{C}_{X_i}$.

Now define
\[
\gamma:\prod_{j=1}^n\Delta X_j \multimap \prod_{j=1}^n\Delta X_j,
\]
by the formula
\[
\gamma(\mu_1,\dots,\mu_n)=\prod_{i=1}^n\gamma_i(\mu_1,\dots,\mu_n).
\]

We verify that $\gamma$ is upper semicontinuous.
Suppose $(\mu,\nu)\in\prod_{j=1}^n\Delta X_j\times\prod_{j=1}^n\Delta X_j$ is a pair such that $\nu\notin\gamma(\mu)$. Then there exists
$i$ and a compact set $K\subset X_i\setminus R_i(\mu_i^\ast)$ with
$\nu_i(K)>0$. Define
\[
O_\nu=\{\alpha\in\prod_{j=1}^n\Delta X_j\mid \alpha_i(K)>0\},
\]
which is an open neighborhood of $\nu$.

By Lemma~\ref{C} and continuity of the tensor product, there exists a neighborhood
$O_\mu$ of $\mu$ such that for every $\alpha\in O_\mu$ we have
$R_i(\alpha_i^\ast)\cap K=\varnothing$. Consequently,
$\beta\notin\gamma(\alpha)$ for all $(\alpha,\beta)\in O_\mu\times O_\nu$,
which shows that $\gamma$ is USC.

Consider the family
\[
\mathcal{C}=\Bigl\{\prod_{i=1}^n C_i\;\Big|\; C_i\in\mathcal{C}_{X_i}\Bigr\}.
\]
This family defines a normal convexity on the compactum
$\prod_{j=1}^n\Delta X_j$, and all its elements are connected.
Therefore, by Theorem~\ref{KA}, the multimap $\gamma$ admits a fixed point
$\mu=(\mu_1,\dots,\mu_n)$.

Finally, for each $i$ we have
\[
\mu_i(X_i\setminus R_i(\mu_i^\ast))=0.
\]
Applying Lemma~\ref{P}, we obtain
\[
\mu_i^\ast\!\left(
\prod_{j\neq i}X_j\setminus\prod_{j\neq i}R_j(\mu_j^\ast)
\right)=0,
\]
which shows that $(\mu_1^\ast,\dots,\mu_n^\ast)$ is an equilibrium under uncertainty.
\end{proof}

Example~\ref{ex} shows that, in general, a Nash equilibrium need not be an equilibrium under uncertainty. We now discuss the converse implication. For this purpose, we require the notion of the density of a possibility capacity.

The notion of density for idempotent measures was introduced in \cite{A}. For each possibility capacity $\nu\in\Delta X$ we associate an upper semicontinuous function $[\nu]:X\to[0,1]$ defined by $[\nu](x)=\nu(\{x\})$, called the \emph{density} of $\nu$. For a closed set $F\subset X$ we have
\[
\nu(F)=\max\{[\nu](x)\mid x\in F\},
\]
so $\nu$ is completely determined by its density.

Conversely, every upper semicontinuous function $f:X\to[0,1]$ with $\max f=1$ defines a possibility capacity $(f)\in\Delta X$ by
\[
(f)(F)=\max\{f(x)\mid x\in F\},
\]
for closed sets $F\subset X$. Clearly, $([\nu])=\nu$ and $[(f)]=f$; see \cite{R2} for details.

It was shown in \cite{R6} that for $\mu_1\in\Delta X_1$, $\mu_2\in\Delta X_2$, and $(x,y)\in X_1\times X_2$ we have
\[
[\mu_1\otimes\mu_2](x,y)=[\mu_1](x)\wedge[\mu_2](y).
\]

\begin{proposition}\label{con}
Let $(f,g):X\times Y\to\mathbb{R}^2$ be a game in pure strategies with compact spaces $X$ and $Y$ and continuous payoff functions $f$ and $g$. If a pair $(\nu,\mu)\in\Delta X\times\Delta Y$ is an equilibrium under uncertainty, then $(\nu,\mu)$ is a Nash equilibrium of the corresponding game in capacities with payoff functions $ef$ and $eg$ defined by the max-plus integral.
\end{proposition}

\begin{proof}
Assume, contrary to the statement, that $(\nu,\mu)\in\Delta X\times\Delta Y$
is an equilibrium under uncertainty but is not a Nash equilibrium of the
corresponding game in capacities. Then, without loss of generality, there
exists $\nu'\in\Delta X$ such that
\[
ef(\nu',\mu)>ef(\nu,\mu).
\]

Let $\nu_1\in\Delta X$ denote the greatest element, i.e.
$\nu_1(A)=1$ for every nonempty closed set $A\subset X$, equivalently
$[\nu_1](x)=1$ for all $x\in X$. Since $\nu_1$ dominates every element of
$\Delta X$, we have
\[
ef(\nu',\mu)\le ef(\nu_1,\mu),
\]
and hence
\[
ef(\nu_1,\mu)>ef(\nu,\mu).
\]

By the definition of the max-plus integral,
\[
ef(\nu_1,\mu)
=
\max\{\ln((\nu_1\otimes\mu)(f_t))+t\mid t\in\mathbb{R}\}.
\]
Using the description of the density of the tensor product, this expression
can be rewritten as
\[
ef(\nu_1,\mu)
=
\max\Bigl\{
\max\{\ln([\nu_1](x)\wedge[\mu](y))\mid (x,y)\in f_t\}+t
\;\Big|\; t\in\mathbb{R}
\Bigr\}.
\]
Since $[\nu_1](x)=1$ for all $x\in X$, we obtain
\[
ef(\nu_1,\mu)
=
\max\Bigl\{
\max\{\ln([\mu](y))\mid \exists\,x\in X \text{ with } (x,y)\in f_t\}+t
\;\Big|\; t\in\mathbb{R}
\Bigr\}.
\]

Consequently, there exist $t_0\in\mathbb{R}$ and $(x_0,y_0)\in f_{t_0}$ such that
\[
\ln([\mu](y_0))+t_0>ef(\nu,\mu).
\]

Since $\nu\in\Delta X$, there exists $x_1\in X$ with $[\nu](x_1)=1$. Then,
by the definition of $ef(\nu,\mu)$,
\[
\begin{aligned}
ef(\nu,\mu)
&=
\max\Bigl\{
\max\{\ln([\nu](x)\wedge[\mu](y))\mid (x,y)\in f_t\}+t
\;\Big|\; t\in\mathbb{R}
\Bigr\} \\
&\ge
\max\Bigl\{
\max\{\ln([\mu](y))\mid (x_1,y)\in f_t\}+t
\;\Big|\; t\in\mathbb{R}
\Bigr\}
= P_1(x_1,\mu).
\end{aligned}
\]

On the other hand, from $(x_0,y_0)\in f_{t_0}$ we obtain
\[
P_1(x_0,\mu)\ge \ln([\mu](y_0))+t_0>ef(\nu,\mu)\ge P_1(x_1,\mu).
\]
Thus $x_1$ is not a best response to $\mu$, i.e. $x_1\notin R_1(\mu)$.
However, since $[\nu](x_1)=1$, we have $\nu(\{x_1\})=1$, which contradicts the
assumption that $(\nu,\mu)$ is an equilibrium under uncertainty.

This contradiction completes the proof.
\end{proof}

We emphasize that Proposition~\ref{con} is only a partial result, since it requires $(\nu,\mu)\in\Delta X\times\Delta Y$. The importance of this case is highlighted by Theorem~\ref{Exi}. The general situation remains open.

\begin{problem}\label{pr}
Let $(f,g):X\times Y\to\mathbb{R}^2$ be a game in pure strategies with compact spaces $X$ and $Y$ and continuous payoff functions $f$ and $g$. Suppose that a pair of capacities $(\nu,\mu)\in MX\times MX$ is an equilibrium under uncertainty. Is $(\nu,\mu)$ necessarily a Nash equilibrium of the corresponding game in capacities with payoff functions $ef$ and $eg$ defined by the max-plus integral?
\end{problem}

\section{Conclusion}

In this paper we investigated non-cooperative games under uncertainty within a framework based on capacities and max-plus integrals. Two equilibrium concepts were analyzed: Nash equilibrium in mixed strategies represented by non-additive measures, and equilibrium under uncertainty, where players choose pure strategies while evaluating outcomes with respect to non-additive beliefs.
For games with compact strategy spaces and continuous payoff functions, we established existence results for both equilibrium notions. In particular, we showed that while Nash equilibria in mixed strategies always exist in the full space of capacities due to the presence of extremal elements, the existence problem becomes nontrivial when attention is restricted to possibility capacities. By employing an idempotent convexity structure and an abstract fixed point theorem, we proved the existence of Nash min-equilibria in this setting.
We further demonstrated that equilibrium under uncertainty can be ensured for games with max-plus payoff functions and possibility-valued beliefs, and that such equilibria can be represented via tensor products of individual possibility capacities. The relationship between Nash equilibrium in capacities and equilibrium under uncertainty was clarified: while the two concepts do not coincide in general, equilibrium under uncertainty implies Nash equilibrium in the case of possibility capacities.
An open problem remains concerning the general relationship between these equilibrium notions beyond the class of possibility capacities.


\end{document}